\documentclass[11pt]{article}
\usepackage{a4wide}
\usepackage[ansinew]{inputenc}
\usepackage{graphicx}
\usepackage{microtype}
\usepackage{amsmath}
\usepackage{amsthm}
\usepackage{amssymb}
\usepackage[rightcaption]{sidecap}
\usepackage{lipsum}

\usepackage[all,arc,curve,color,frame]{xy}

\newtheorem{thm}{Theorem}
\newtheorem{lem}[thm]{Lemma}
\newtheorem{cor}{Corollary}
\newtheorem{prop}[thm]{Proposition}
\newtheorem{clm}[thm]{Claim}
\newtheorem{probl}{Problem}

\theoremstyle{remark}
\newtheorem{rem}{Remark}

\let\geq\geqslant
\let\leq\leqslant
\newcommand{\joi}{\ast}
\newcommand{\lin}[1]{\overline{#1}} 
\newcommand{\trop}[1]{\hat{#1}} 
\newcommand{\lenv}[1]{{#1}_{\mathrm{le}} } 
\newcommand{\henv}[1]{{#1}_{\mathrm{he}} } 
\newcommand{\mon}[1]{\mathrm{mon}(#1)} 
\newcommand{\msup}[1]{X_{#1} } 
\newcommand{\Sup}[1]{\mathrm{sup}(#1) } 
\newcommand{\Min}[1]{\mathrm{min}(#1) } 
\newcommand{\MMin}[1]{\mathrm{Min}(#1) }
\newcommand{\m}[1]{|\Min{#1}|} 
\newcommand{\prd}{\times}
\newcommand{\pr}{} 

\newcommand{\CC}{\mathrm{C}} 

\newcommand{\T}{\mathrm{T}}
\newcommand{\C}{\mathrm{C}_{\mbox{\tiny $0/1$}}} 
\newcommand{\dec}{\mathrm{D}} 
\newcommand{\dect}{\mathrm{D}_{\mathrm{thr}}}
\newcommand{\RR}{\mathbb{R}} 
\newcommand{\NN}{\mathbb{N}} 
\newcommand{\F}{\mathrm{F}} 
\newcommand{\klase}[1]{\mathsf{#1}} 
\newcommand{\PERM}[1]{\klase{Per}_{#1}} 

\newcommand{\PPPERM}[1]{\klase{Per}^*{}}
\newcommand{\HP}[1]{\klase{Ham}_{#1}} 
\newcommand{\tR}[1]{\Delta_{#1}} 
\newcommand{\ISOL}[1]{\klase{Isol}_{#1}}
\renewcommand{\PATH}[1]{\klase{Path}_{#1}}
\newcommand{\ALLPATH}[1]{f_{#1}}
\newcommand{\Clique}[1]{\klase{Clique}_{#1}}

\newcommand{\ddeg}[2]{\#_{#2}(#1)}

\newcommand{\llength}{l}
\newcommand{\length}[1]{\llength(#1)}

\begin{document}

\title{Lower Bounds for Monotone Counting Circuits\thanks{Research supported by the DFG grant SCHN~503/6-1.}}

\author{Stasys~Jukna\thanks{Institute of Computer Science, Goethe University, Frankfurt am Main, Germany. Affiliated with Institute of Mathematics and Informatics, Vilnius University, Vilnius, Lithuania. Email: jukna@thi.informatik.uni-frankfurt.de}}

\maketitle
\date{}

\begin{abstract}
A $\{+,\prd\}$-circuit \emph{counts} a given multivariate polynomial $f$, if its values on
$0$-$1$ inputs are the same as those of $f$; on other inputs the circuit may output arbitrary values. Such a circuit counts the number of monomials of $f$  evaluated to $1$ by a given $0$-$1$ input vector (with multiplicities given by their coefficients).
A circuit \emph{decides} $f$ if it has the same $0$-$1$ roots as $f$.
We first show that some multilinear polynomials can be exponentially easier to count than to compute them, and can be exponentially easier to decide than to count them. Then we give general lower bounds on the size  of counting circuits.\\[2ex]

\noindent{\bf Keywords:} arithmetic circuits, boolean circuits, counting complexity, lower bounds
\end{abstract}

\section{Introduction}

In this paper we consider computational complexity of multivariate
polynomials with nonnegative integer coefficients:
\begin{equation}\label{eq:pol}
f(x_1,\ldots,x_n)=\sum_{e\in\NN^n}c_e\prod_{i=1}^n x_i^{e_i}\,,
\end{equation}
where $c_e\in\NN=\{0,1,2,\ldots\}$, and $x_i^0=1$. Products
$\prod_{i=1}^n x_i^{e_i}$ are \emph{monomials} of $f$; we will often omit monomials whose coefficients $c_e$ are zero. The polynomial is \emph{multilinear},
if $c_e=0$ for all $e\not\in\{0,1\}^n$, and is \emph{homogeneous} of degree $d$,
if $e_1+\cdots+e_n=d$ for all $e$ with $c_e\neq 0$.

A standard model of compact representation of such polynomials (with nonnegative coefficients) is that of
monotone arithmetic circuits, i.e. of $\{+,\prd\}$-circuits. Such a circuit
is a directed
acyclic graph with three types of nodes: input,
addition ($+$), and multiplication  ($\prd$). Input nodes have
fanin zero, and correspond to variables $x_1,\ldots,x_n$. All other nodes have fanin
two, and are called \emph{gates}. The \emph{size} of a circuit is
the number of gates in it.

Every $\{+,\prd\}$-circuit syntactically
\emph{produces}  a unique monotone polynomial $F$ with nonnegative integer coefficients in a natural way:
 the polynomial produced at an input gate $x_i$ consists of a single
monomial $x_i$, and the polynomial produced at a sum (product) gate is the sum (product) of
polynomials produced at its inputs; we use distributivity to write a
product of polynomials as a sum of monomials.
The polynomial $F$ produced by the circuit itself is the polynomial produced at its
output gate.
Given a polynomial $f(x_1,\ldots,x_n)$, we say that the circuit
\begin{itemize}
\item \emph{computes} $f$, if $F(a)=f(a)$ holds for all $a\in\NN^n$, where $\NN=\{0,1,2,\ldots\}$;
\item \emph{counts}  $f$, if $F(a)=f(a)$ holds for all $a\in\{0,1\}^n$;
\item \emph{decides}  $f$, if $F(a)=0$ exactly when $f(a)=0$ holds for all $a\in\{0,1\}^n$.
\end{itemize}
In this paper we are mainly interested in $\{+,\prd\}$-circuits \emph{counting} a given polynomial $f$. Such a circuit needs only to correctly compute the restriction
$f:\{0,1\}^n\to \NN$ of $f$ on $0$-$1$ inputs.
Note that,
if the polynomial $f$ is monic (has no coefficients $>1$) then,  on every $0$-$1$
input $a\in\{0,1\}^{n}$, $f(a)$ is the number of monomials of $f$ satisfied by (evaluated to $1$ on)
$a$. For example, in the case of the permanent polynomial
\[
\PERM{n}(x)=\sum_{h}\prod_{i=1}^nx_{i,h(i)}
\]
with the summation over all permutations $h$ of $[n]=\{1,\ldots,n\}$, the value $\PERM{n}(a)$ is the number of
perfect matchings in the bipartite $n\times n$ graph $G_a$ specified by input~$a\in\{0,1\}^{n\times n}$;
nodes $i$ and $j$ are adjacent in $G_a$  if and only if $a_{ij}=1$. Thus, a circuit counting
$\PERM{}$ outputs the number of perfect matchings in $G_a$,
whereas a circuit deciding
this polynomial merely tells us whether $G_a$ contains a perfect matching.

\begin{rem}
Let us stress that we only consider \emph{monotone} arithmetic circuits. The reason is that
counting $\{+,-,\prd\}$-circuits are already omnipotent: they are as powerful as boolean
$\{\lor,\land,\neg\}$-circuits. This is because each boolean operation can be simulated over $\{0,1\}$:
$x\land y$ by $x\pr y$, $\neg x$ by $1-x$, and $x\lor y$ by $x+y-xy$.
\end{rem}

If a $\{+,\prd\}$-circuit computes, counts or only decides a given polynomial $f$, what can then be said about the structure of the produced by the circuit polynomial $F$?

If the circuit
\emph{computes} $f$, then $F=f$ must hold, that is, then
the produced polynomial $F$ and the target polynomial $f$ must
coincide as formal expressions, i.e. as sums of monomials (see, e.g. Claim~\ref{clm:uniq} below for simple a proof).
In particular, then
$\mon{F}=\mon{f}$
must also hold, where
\begin{itemize}
\item $\mon{f}$ is the set of monomials appearing
in $f$ with nonzero coefficients.
\end{itemize}
This ensures that no ``invalid'' monomials can be formed
during the computation, and severely limits the power of such
circuits. In particular, if the target polynomial $f$ is
\emph{multilinear} (no variable has degree larger than $1$, then the
circuit itself must be multilinear: the polynomials produced at inputs
of each product gate must
depend on \emph{disjoint} sets of variables.  This limitation is
essentially exploited in all lower bounds for monotone arithmetic
circuits, including \cite{schnorr,shamir,jerrum,valiant80,snir,gashkov,tiwari,GS}.

\begin{SCfigure}[10]
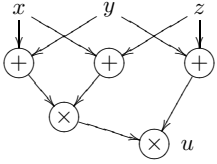

 \scalebox{0.8}{
 \xygraph { !{<0cm,0cm>;<1cm,0cm>:<0cm,-0.9cm>::}
    !{(0,0)}*+{x}="x"
    !{(1.5,0)}*+{y}="y"
    !{(3,0)}*+{z}="z"
    !{(0,1)}*+=[o]+[F]{+}="s1"
    !{(1.5,1)}*+=[o]+[F]{+}="s2"
    !{(3,1)}*+=[o]+[F]{+}="s3"
    !{(0.75,2)}*+=[o]+[F]{\prd}="p1"
    !{(2.25,2.5)}*+=[o]+[F]{\prd}="p2"
    !{(2.8,2.5)}*+={u}="u"
    "x":"s1" "x":"s2" "y":"s1" "y":"s3" "z":"s2" "z":"s3" "s1":"p1"
    "s2":"p1" "s3":"p2" "p1":"p2"
    }
    }
    \caption{A circuit of size $5$ counting the polynomial $f=2xyz+2xy+2xz+2yz$, and deciding the
    polynomial $g=xy+xz+yz$. The circuit itself produces the polynomial $F=(x+y)(y+z)(x+z)=2xyz+x^2y+xy^2+x^2z+xz^2+y^2z+yz^2$. Gate $u$ is the output gate.}
  \label{fig:circuit}
\end{SCfigure}

In counting circuits, $\mon{F}=\mon{f}$ needs not to hold, due to
the multiplicative idempotence axiom $x^2=x$ valid on $0$-$1$ inputs. That is, here exponents (and hence, degrees
of monomials) do not mater (see Fig.~\ref{fig:circuit}). Still, it can be shown (see Lemma~\ref{lem:monic} below) that here we have a weaker, but still strong enough property $\Sup{F}=\Sup{f}$, where
\begin{itemize}
\item $\Sup{f}$ is the \emph{support} of $f$,
that is, the family of sets of variables of monomials in~$\mon{f}$.
\end{itemize}

In deciding circuits, even $\Sup{F}=\Sup{f}$ needs not to hold, due to the additional absorption axiom $x+xy=x$. In such circuits, we only have a weak property $\Min{F}=\Min{f}$, where
\begin{itemize}
\item
$\Min{f}\subseteq \Sup{f}$ is the family of all members of $\Sup{f}$
which are minimal in the sense than they do not contain any other members of $\Sup{f}$;
hence, $\Min{f}$ forms an antichain.
\end{itemize}
Deciding $\{+,\prd\}$-circuits are actually monotone \emph{boolean} circuits, and we have the following relations concerning the minimum circuit size for every given polynomial (we will prove that both gaps can be exponential):
\[
\mbox{Deciding $\leq$ Counting $\leq$ Computing.}
\]
To prove lower bounds for deciding, and hence, also for counting  $\{+,\prd\}$-circuits, one can use lower-bounds arguments
for monotone boolean circuits (see, e.g. \cite[Chapt. 9]{myBFC-book} and the literature cited herein), but these are not easy to apply.
The reason here lies in a ``dual character'' of these arguments: in order to obtain a large lower bound on the decision complexity of a given polynomial $f$, not only the set of monomials of the polynomial $f$ itself but also that of the ``dual'' polynomial $f^*$ must have some good structural properties (see the discussion before Theorem~\ref{thm:main} below).

On the other hand, due to the limitations we mentioned above, lower bounds for
$\{+,\prd\}$-circuits \emph{computing} a given polynomial are
relatively easy to obtain, because here we have a full knowledge about the polynomial which a circuit must produce. In particular, there is then no need to consider dual polynomials. \emph{Counting} $\{+,\prd\}$-circuits allow more freedom, because they
can use $x^2=x$. In this case we only know the structure of the support of the produced polynomial, but not about its monomials. So, it is natural to ask
whether known lower bounds for exactly computing $\{+,\prd\}$-circuits can be extended to counting circuits?

That they sometimes \emph{can} be extended was demonstrated by Sengupta and
Venkateswaran in \cite{sengupta}, where they show that an exponential
lower bound of Jerrum and Snir \cite{jerrum} for $\{+,\prd\}$-circuits
computing the permanent polynomial $\PERM{}$ can be adopted to yield
the same lower bound for circuits only counting this
polynomial. Still, at least three questions remained open:
\begin{enumerate}
\item Can counting circuits be substantially
  smaller than computing circuits?

  \item Can deciding circuits be substantially
  smaller than counting circuits?

\item Can lower-bounds \emph{arguments} for computing $\{+,\prd\}$-circuits, not just bounds for specific
  polynomials (like the permanent polynomial), be extended to $\{+,\prd\}$-counting circuits?
\end{enumerate}
In this paper, we answer these questions affirmatively.

\section{Results}

For a polynomial $f$, let $\CC(f)$ denote the minimum size of a
$\{+,\prd\}$-circuit \emph{computing} $f$,  $\C(f)$ the minimum size of such a circuit
\emph{counting} $f$, and  $\dec(f)$ the minimum size of a
$\{+,\prd\}$-circuit \emph{deciding} $f$.
Note that, for every polynomial $f$, we have that
\[
\dec(f)\leq \C(f)\leq \CC(f)\,.
\]
We first show that the gaps $\CC(f)/\C(f)$ as well as $\C(f)/\dec(f)$ can be exponential.
When doing this, we will use known lower bound for the permanent polynomial.

\begin{thm}[\cite{jerrum,sengupta}]\label{thm:perm}
If $f=\PERM{n}$, then $\C(f)\geq n2^{n-1}$.
\end{thm}
This lower bound on $\CC(f)$ was proved by Jerrum and Snir \cite{jerrum},
and was extended to $\C(f)$ by
Sengupta and Venkateswaran \cite{sengupta} (see also Corollary~\ref{cor:bounds1} below for a short
proof of a weaker $2^{\Omega(n)}$ lower bound).

We will also use the (simple) fact that it is not harder to compute the so-called ``lower'' and
``higher'' envelopes of polynomial than to compute the polynomial itself.
The \emph{lower envelope} of a polynomial $f$ is a homogeneous polynomial
$\lenv{f}$ consisting of the monomials of $f$ of smallest degree. The \emph{higher envelope}
$\henv{f}$ is defined by taking monomials of largest degree. (As
usually, the degree of a monomial is the sum of exponents of its
variables, and a polynomial is \emph{homogeneous}, if all its monomials have the same degree.) As observed already by Jerrum and Snir \cite{jerrum}, every
$\{+,\prd\}$-circuit producing a polynomial $f$ can be easily
transformed into a circuit producing $\lenv{f}$ or $\henv{f}$ by just discarding (if
necessary) some of the sum-gates. Hence, we always have
\begin{equation}\label{eq:env}
\CC(f)\geq \max\left\{\CC(\henv{f}), \CC(\lenv{f}\right\}\,.
\end{equation}

\subsection{Gaps}

To show that the gap $\CC(f)/\C(f)$ can be exponential, we will show a stronger fact that both
gaps $\C(\henv{f})/\C(f)$ and $\C(\lenv{f})/\C(f)$ can be exponential. Recall that,
by \eqref{eq:env}, no such gap is possible for
\emph{computing} $\{+,\prd\}$-circuits.

\begin{thm}\label{thm:main0}
There are multilinear polynomials $f$ and $g$ of $n$ variables such that
$\C(f)=O(n)$ and $\C(g)=O(n^{3/2})$, but
 both
$\C(\henv{f})$ and $\C(\lenv{g})$ are $2^{\Omega(\sqrt{n})}$.
\end{thm}

\begin{rem} Together with \eqref{eq:env}, the theorem implies that the gap $\CC(f)/\C(f)$ between the sizes of $\{+,\prd\}$-circuits computing and counting $f$ can be exponential.
Important in this result is that the gap is obtained for
\emph{multilinear} polynomials: this shows that, under the presence of
multiplicative idempotence $x^2=x$, non-multilinear circuits counting
multilinear polynomials can be much more efficient.
In this connection, let us mention that without this restriction (to multilinear
polynomials) a non-trivial gap follows from the classical lower bound
$\Omega(n\log d)$ of Strassen~\cite{strassen1}, and Baur and
Strassen~\cite{BS83} on the size of arithmetic (not necessarily
monotone) circuits computing the polynomial
$f=x_1^d+x_2^d+\cdots+x_n^d$, which can be trivially counted by a
$\{+,\prd\}$-circuit $F=x_1+x_2+\cdots+x_n$ of size only $n-1$.
But this example merely says that, under the presence of
multiplicative idempotence $x^2=x$, rising to powers is redundant.
\end{rem}

To show that the gap $\C(f)/\dec(f)$ can also be exponential, it is enough to take
any polynomial $g(x_1,\ldots,x_n)$ such that $\C(g)$ is exponential,
and consider the polynomial $f=g+h$ where $h=\sum_{i=1}^n x_i$. If $g(0,\ldots,0)=0$ then, on every $0$-$1$ input $a$, we have that $f(a)=0$ if and only if $h(a)=0$. So, $f$ has a small
decision complexity: $\dec(f)\leq \dec(h)\leq n$. So, if the counting complexity $\C(f)$
of the extended polynomial $f$ remains exponential, then the gap $\C(f)/\dec(f)$ is exponential.
In particular, one can establish such a gap by using the permanent
polynomial $g=\PERM{}$ (the only small ``technicality'' here is to show that
the counting complexity of $f$ remains large).

\begin{thm}\label{thm:main01}
If $f=\PERM{n}+\sum_{i,j=1}^n x_{ij}$,
then $\dec(f)\leq n^{2}$ but $\C(f)=2^{\Omega(n)}$.
\end{thm}

The polynomial used in this theorem is somewhat artificial.
 Actually, one can establish an exponential gap using a more natural (and important) $s$-$t$ path polynomial $\PATH{n}$.
 This polynomial has one variable $x_{i,j}$ for each edge of a complete undirected graph on $n+2$ nodes $\{s,1,2,\ldots,n,t\}$.  Each monomial of $f$ corresponds to a simple directed path from node $s$ to node $t$:
\[
\PATH{n}(x)=x_{s,t}+\sum_{l=1}^n\sum_{\substack{i_1,\ldots,i_l\\
 \mathrm{distinct}} } x_{s,i_1}x_{i_1,i_2}\cdots x_{i_{l-1},i_l}x_{i_l,t}\,.
\]
On a $0$-$1$ input $a$, $\PATH{n}(a)$ gives the number of $s$-$t$ paths in the graph specified by $a$.
Jerrum and Snir~\cite{jerrum} have shown that every $\{+,\prd\}$-circuit \emph{computing} $f=\PATH{n}$
must have exponential size, i.e. that $\CC(f)=2^{\Omega(n)}$.
We show that even $\{+,\prd\}$-circuits \emph{counting} $\PATH{}$
must have exponential size.

\begin{thm}\label{thm:main02}
If $f=\PATH{n}$, then $\dec(f)=O(n^3)$, but
$\C(f)\geq 2^{n^{\Omega(1)}}$.
\end{thm}

\subsection{Lower bounds}

Recall that, if a $\{+,\prd\}$-circuit \emph{computes} a given polynomial $f$,
then the produced by the circuit polynomial $F$ must just coincide with $f$
(as formal expressions). In counting and deciding circuits we only have weaker conditions on~$F$.

By the \emph{linearization} of a polynomial $f$ we will mean a multilinear polynomial
$\lin{f}$ obtained from $f$ by removing all (nonzero) exponents from all monomials of $f$. For example,
the linearization of $f=2xy^2+ 3x^4y^2 + 6y^2z$ is $\lin{f}=5xy + 6yz$.
It is clear that $\lin{f}(a)=f(a)$ holds for all $a\in\{0,1\}^n$.

\begin{lem}\label{lem:monic}
If a $\{+,\prd\}$-circuit producing a polynomial $F$ counts
$f$, then $\lin{F}=\lin{f}$, and hence, also $\Sup{F}=\Sup{f}$. A $\{+,\prd\}$-circuit
decides $f$ if and only if $\Min{F}=\Min{f}$.
\end{lem}

Our next structural result is the following lemma.
The \emph{support} of a monomial is the set of variables appearing in it with nonzero degree; the size of this set is the \emph{length} of the monomial.  A product $g\pr h$ of two polynomials is \emph{$m$-balanced}, if  the minimum length $l$ of one these polynomials
satisfies  $m/3<l\leq 2m/3$.
A monomial $p$ \emph{appears} $m$-\emph{balanced} in a product $g\pr h$ of two polynomials,
if there are monomials $r\in\mon{g}$ and $s\in\mon{h}$ such that
$rs$ and $p$ have the same support, and the length $l$ of $r$ satisfies
$m/3<l\leq 2m/3$.
Note that here the order of polynomials in their product $g\pr h$ is
important: the condition is only on parts of monomials appearing in the first polynomial.
In particular, if several monomials appear $m$-balanced in $g\pr h$, then we know the bounds on the lengths of their parts in \emph{one and the same} of the two polynomials.

\begin{lem}\label{lem:balan}
Let $m\geq 2$, and let $f$ a polynomial of counting complexity $\C(f)=s$.
\begin{itemize}
\item[\mbox{\rm (i)}] If every monomial of $f$ has length at least $m$, then $\Sup{f}$ is a union of at most $s$ supports of $m$-balanced products of polynomials.
\item[\mbox{\rm (ii)}] There are $s$  products $g\pr h$ of polynomials such that $\Sup{g\pr h}\subseteq \Sup{f}$, and
every monomial of $f$ of length at least $m$
appears $m$-balanced in at least one of these products.
\end{itemize}
\end{lem}

Various versions of claim (i) (with degree of or the total number of variables in polynomials used instead of their length) were observed by several authors including Hyafil~\cite{hyafil}, Jerrum and Snir~\cite{jerrum}, Valiant~\cite{valiant80}, and Raz and Yehudayoff~\cite{RY}. The advantage of claim (ii) is its wider applicability: the polynomial $f$ itself is allowed to have also short monomials, shorter than~$m$.

Our next results are more \emph{explicit} lower bounds for counting circuits.
The \emph{$r$-th degree}, $\ddeg{A}{r}$, of a family of sets $A$ is
the maximum number of sets in $A$
containing a fixed $r$-element set:
\[
\ddeg{A}{r}=\max_{|b|=r} |\{a\in A\colon a\supseteq b\}|\,.
\]
In other words, the intersection of any $\ddeg{A}{r}$ sets in $A$ can have at most $r$ elements.
Note that
\[
|A|=\ddeg{A}{0}\geq \ddeg{A}{1}\geq \ldots\geq \ddeg{A}{r}=1> 0=\ddeg{A}{r+1}\,.
\]
where  $r=\max\{|a|\colon a\in A\}$. Also, $A\subseteq B$ implies $\ddeg{A}{r}\leq \ddeg{B}{r}$. If $A$ is a graph (viewed as a set of edges), then $\ddeg{A}{1}$
is the maximum degree of $A$.
In general, $\ddeg{A}{r}$ is related with $|A|$ as follows: if
 $A$ is a family of $m$-element subsets on $[n]$, then for every $r\leq m$ we have that
\[
|A|\binom{m}{r}\leq \ddeg{A}{r}\cdot \binom{n}{r}\,.
\]
This can be shown by counting  in two ways the number $M$ of pairs $(a,b)$,
where $a\in A$,
$|b|=r$ and $a\supseteq b$ holds. By first fixing sets $a\in A$, we get that $M$ is
equal to the left-hand side. By fixing sets $b$, and taking all possible $m$-element sets $a$ containing $b$, we get that $M$ is at most the right-hand size.

As we mentioned in the introduction, lower bounds for deciding, and hence, also for counting  $\{+,\prd\}$-circuits, can be obtained using lower-bounds arguments
for monotone boolean circuits (see, e.g. \cite[Chapt. 9]{myBFC-book} and the literature cited herein), but these are not easy to apply.
The reason here lies in a ``dual character'' of these arguments: in order to obtain a large lower bound
of the decision complexity of a polynomial $f$ given by \eqref{eq:pol},
not only the set of monomials of $f$ itself but also that of its ``dual'' $f^*$
must have some good structural properties. The \emph{dual} $f^*$ of a polynomial
\[
f=\sum_{u\subseteq [n]}c_u\prod_{i\in u}x_i\ \ \mbox{ is }\ \ f^*=\prod_{u:c_u>0}
\sum_{i\in u}x_i\,.
\]
Note that, for every $0$-$1$ input $a=(a_1,\ldots,a_n)$, $f(a)= 0$ if and only if $f^*(\bar{a})\neq 0$,
where $\bar{a}=(1-a_1,\ldots,1-a_n)$. This holds, because every set in $\Sup{f^*}$
intersects every set in $\Sup{f}$.
More precisely, a general lower bound for deciding $\{+,\prd\}$-circuits is
the following.

\begin{thm}[\cite{Juk99}]
Let $f(x_1,\ldots,x_n)$ be a polynomial, and $2\leq r,s\leq n$ be integers. Then for every $A\subseteq \Sup{f}$ and $B\subseteq \Sup{f^*}$ such that $\ddeg{A}{1}\leq |A|/2(s-1)$,
we have
\[
\dec(f)\geq \min\left\{\frac{|A|}{2(s-1)^r\cdot \ddeg{A}{r}},\
\frac{|B|}{(r-1)^s\cdot \ddeg{B}{s}} \right\}\,.
\]
\end{thm}

As shown in \cite{Juk99} (see also \cite[Chapt. 9]{myBFC-book}), this criterion
allows to obtain strong (super-polynomial) lower bounds on $\dec(f)$, and hence, also on $\C(f)$ and $\CC(f)$,  for some
explicit polynomials. The strength of this criterion lies in the possibility to arbitrarily chose both the
parameters $r,s$ as well as sub-families $A$ and $B$.
The weakness, however, lies in its
 ``dual nature'' making it not easy to apply: \emph{both} $|A|/\ddeg{A}{r}$
\emph{and} $|B|/\ddeg{B}{s}$ must be large. It is usually easy to ensure that
$|A|/\ddeg{A}{r}$ is large. The problem, however, is with the dual set $B$, because
the set of monomials of the dual polynomial $f^*$ may be rather ``messy'', even though
the polynomial $f$ itself has a ``nice'' structure. Say, if $f=\PERM{n}$, then
$|A|/\ddeg{A}{r}=n!/(n-r)!$ is large enough already for $A=\Sup{f}$.
But monomials of $f^*$ correspond then to complements of graphs without perfect matchings,
and it is difficult to ensure that  $|B|/\ddeg{B}{s}$ is also large for some family $B$
of such graphs.

For counting $\{+,\prd\}$-circuits, we have a much more handy lower-bounds
criterion, avoiding the need of dual polynomials. By the $r$-th degree, $\ddeg{f}{r}$, of a polynomial $f$ we will mean
the $r$-th degree $\ddeg{A}{r}$ of its support $A=\Sup{f}$.
Thus, if $f$ is multilinear, then $\ddeg{f}{r}$ is the maximum number of monomials of $f$
containing a common factor of degree~$r$.

\begin{thm}\label{thm:main}
Let $f=g+h$ be a polynomial such that every monomial of $g$ has at least $m\geq 2$ variables, and every monomial of $h$ has fewer than $m/3$ variables.  Then there is an integer $r$ between $m/3$
 and $2m/3$ such that
\begin{equation}\label{eq:free1}
    \C(f)\geq \frac{|\Sup{g}|}{\ddeg{g}{r}\cdot\ddeg{g}{m-r}}\,.
\end{equation}
\end{thm}

There is yet another general lower-bounds criterion for monotone arithmetic circuits, due
to Gash\-kov~\cite{gashkov}, and Gashkov and Sergeev~\cite{GS}. They call a polynomial $f$ $(k,l)$-\emph{sparse}, if
\[
\mbox{$\mon{g\pr h}\subseteq \mon{f}$ implies  $|\mon{g}|\leq k$
or $|\mon{h}|\leq l$.}
\]
They proved that  $\CC(f)+1\geq |\mon{f}|/\max\{k^3,l^2\}$ holds for every such polynomial.
Note that the bound is not trivial, because the fact that $|\mon{g}|\leq k$
or $|\mon{h}|\leq l$ holds does not imply that $|\mon{g\pr h}|\leq kl$ must also hold (because
we have an ``or'', not ``and'' here).
To obtain a similar lower bound for counting circuits, we will modify their notion of ``sparsity''.

Let, as before, $\Min{f}\subseteq \Sup{f}$ denote the family of all members of $\Sup{f}$
which are minimal in the sense than they do not contain any other members of $\Sup{f}$. Call a polynomial $f$ $(k,l)$-\emph{free}
if, for every two polynomials $g$ and $h$,
\[
\mbox{$\Sup{g\pr h}\subseteq \Sup{f}$ implies $\m{g}\leq k$ or $\m{h}\leq l$.}
\]
The reason to only require $|\Min{g}|\leq k$ instead of $|\Sup{g}|\leq k$
 is that then it is (potentially) easier to show that a given polynomial is $(k,l)$-free: $|\Min{g}|$ can be much smaller than $|\Sup{g}|$.

\begin{thm}\label{thm:main2}
Let $1\leq k\leq l$ be integers.
For every $(k,l)$-free polynomial $f$, its support $\Sup{f}$ is a union of
at most $2\C(f)$ supports $\Sup{g\pr h}$ of products $g\pr h$ of polynomials such that
$|\Min{g\pr h}|\leq kl^2$. In particular,
\[
\C(f)\geq \frac{|\Min{f}|}{2kl^2}\,.
\]
\end{thm}

\begin{rem}
The proofs of Theorems~\ref{thm:main} and \ref{thm:main2} extend to $\C(f)$
the arguments used in \cite{gashkov,GS,juk14} to lower-bound $\CC(f)$.
The main difficulty with the extension (stipulated by the idempotence axiom $x^2=x$) is that, unlike the measure $\mu(f)=|\mon{f}|$ (used to lower-bound $\CC(f)$), the measures $|\Sup{f}|$ and $|\Min{f}|$
are no more ``monotone '' in the sense that $\mu(f)\leq \mu(f\pr g)$. To see this, take, for example,
$f=x_1+x_2+\cdots+x_n$ and $g=x_1x_2\cdots x_n$. Then $|\Sup{f}|=n$ but $|\Sup{f\pr g}|=1$.
\end{rem}

\begin{rem}
The proofs of Theorems~\ref{thm:main} and \ref{thm:main2} are based on the fact
(Lemma~\ref{lem:monic}) that, if a $\{+,\prd\}$-circuit counts a polynomial $f$,
then the produced by the circuit polynomial $F$ must satisfy $\Sup{F}=\Sup{f}$.
Thus, these bounds do not extend to monotone \emph{boolean} circuits,
where we only have a much weaker property $\Min{F}=\Min{f}$.
\end{rem}

\section{Some Applications}
\label{sec:appl}

Theorem~\ref{thm:main} allows us to easily obtain strong lower bounds on
$\C(f)$ for many polynomials.
Let us demonstrate this on some of
them. First, associate with very set $H$ of permutations $h:[n]\to[n]$ the
polynomial in $n^2$ variables $x_{i,j}$:
\[
f_H(x) = \sum_{h\in H}\prod_{i=1}^n x_{i,h(i)}\,.
\]
For example, if $H$ consists of all permutations, then $f_H$
is the permanent polynomial $\PERM{n}$. If $H$ consists of al \emph{cyclic} permutations,
then the monomial of $f_H$ correspond to Hamiltonian cycles in $K_n$.

\begin{cor}\label{cor:bounds1}
For every set $H$ of permutations of $[n]$, there is an $r$ such that
$n/3<r\leq 2n/3$ and
\[
\C(f)\geq \frac{|H|\binom{n}{r}}{n!}\,.
\]
\end{cor}

In particular, $\C(\PERM{n})\geq \binom{n}{r}=2^{\Omega(n)}$.

\begin{proof}
  The polynomial $f_H$ has $|H|$ monomials, each
   specified by a permutation $h\in H$ of $[n]$. If some $r$ variables are fixed, this fixes $r$ values of $h$.
   Hence, at most $(n-r)!$ of the permutations can
  take $r$ pre-described values, implying that  $\ddeg{f}{r}\leq (n-r)!$.  Thus, Theorem~\ref{thm:main} gives that $\C(f)$ is at least $|H|$ divided
   by the maximum of $r!(n-r)!$ over all $n/3<r\leq 2n/3$.
\end{proof}

In some cases, Theorem~\ref{thm:main} allows to even obtain almost optimal bounds.
 A \emph{partial $t$--$(n,m,\lambda)$ design} is a family $A$ of
  $m$-element subsets of $\{1,\ldots,n\}$ such that any $t$-element
  set is contained in at most $\lambda$ of its members. We can
  associate with each such design $A$ a multilinear polynomial
  \[
  f_A(x)=\sum_{a\in A}\prod_{i\in a}x_i\,.
  \]
\begin{cor}\label{cor:design}
For every partial $t$--$(n,m,\lambda)$ design $A$ with $m/3\leq t\leq 2m/3$, we have
$\C(f_A)\geq |A|/\lambda^2$.
\end{cor}

\begin{proof}
For all $m/3\leq r\leq 2m/3$, we have that both $r$ and $m-r$ are at least $m/3$.
Thus, the design property implies that both $\ddeg{A}{r}$ and $\ddeg{A}{m-r}$
are at most $\lambda$, and the desired lower bound follows directly from
Theorem~\ref{thm:main}.
\end{proof}

There are many explicit partial designs $A$ with $\lambda << \sqrt{|A|}$. For every of them, the counting complexity of the polynomial $f_A$ is almost the same as the number
of monomials. To give an example, let $n$ be a prime power, and let $A$ consist of all
subsets $a=\{(i,h(i)\colon i\in \mathrm{GF}(n)\}$ of the grid $\mathrm{GF}(n)\times \mathrm{GF}(n)$ corresponding to
polynomials  $h(z)$ of degree at most $d-1$ over $\mathrm{GF}(n)$.
Since no two distinct polynomials of degree $<d$ can coincide on $d$ points, we have that
no two monomials of $f$ can share $d$ variables in common, $A$ is a partial
$1$-$(n^2,n,1)$ design, and we obtain $n^d=|A|\leq \C(f_A)\leq n^{d+1}$.

Theorem~\ref{thm:main2} is more difficult to apply than Theorem~\ref{thm:main}, but it may help for polynomials, on which the latter theorem fails. To demonstrate this,
let $A$ be a set of edges of a bipartite point-line incidence graph of a projective
plane $PG(2,q)$, introduced by Singer~\cite{singer}. The nodes on the left-side correspond to $n=q^2+q+1$ points $x$,
and those on the left-side to $n$ lines $L$, and $x$ and $L$ are adjacent if $x\in L$.
Since every line $L$ has $|L|=q+1$ points, and every point lies in $q+1$ lines, this is a $d$-regular graph of degree
$d=q+1>\sqrt{n}$. Moreover, the graph is $K_{2,2}$-free (i.e. contains no complete $2\times 2$ subgraphs), because every two point lie in only one line, and every two lines share only one point. For the polynomial
\[
f_A(x)=\sum_{uv\in A}x_ux_v\,,
\]
Theorem~\ref{thm:main}
can only give a trivial lower bound $\C(f_A)\geq |A|/d^2=\Omega(\sqrt{n})$. Indeed,
in this case we have $m=2$, and hence, $r=1$. But then both $\ddeg{f_A}{r}$ and
$\ddeg{f_A}{m-r}$ are equal $d>\sqrt{n}$.
On the other hand, it is not difficult to verify that the $K_{2,2}$-freeness of $A$ implies that the polynomial $f_A$ is $(k,l)$-free
for $k=l=1$. Thus, Theorem~\ref{thm:main2} yields an almost optimal
lower bound
\[
\C(f_A)=\Theta(n^{3/2})\,.
\]
As a second example, let us consider the structurally much simpler
 \emph{triangle polynomial} of $n=3m^2$ variables with $m^{3}=\Theta(n^{3/2})$ monomials:
\[
\tR{n}(x,y,z)=\sum_{i,j,k\in[m]} x_{ik}y_{kj}z_{ij}\,.
\]
Schnorr~\cite{schnorr} has shown that $\CC(\tR{n})=\Theta(n^{3/2})$;
this also follows from the lower bound of Gashkov and Sergeev~\cite{GS}
mentioned above, because the polynomial is $(1,1)$-sparse: any triangle is uniquely determined by any two of its edges.

Concerning \emph{counting} circuit complexity of $f=\tR{n}$,
Theorem~\ref{thm:main} can only yield a trivial lower bound $\C(f)\geq m^3/m=m^2=n/3$,
 because up to $m$ triangles
can share a common edge. Still, Theorem~\ref{thm:main2} (with some more effort) allows us to obtain an almost optimal
lower bound.

\begin{cor}\label{cor:triangle}
  If $f=\tR{n}$, then $\C(f)=\Theta(n^{3/2})$.
\end{cor}

\begin{proof}
   The upper bound $\C(f)=O(m^3)=O(n^{3/2})$ is trivial.  To prove the lower bound $\C(f)=\Omega(m^3)$,
   we will use Theorem~\ref{thm:main2}. Since
   $|\Sup{f}|=m^3$, it is enough to show that $f$ is  $(1,1)$-free.
  To show this, assume that
  $\Sup{g\pr h}\subseteq \Sup{f}$ for some polynomials $g$ and $h$ such that
  $|\Min{g}|\geq 2$ and $|\Min{h}|\geq 2$. Take any two sets $a_1,a_2\in \Min{g}$, and two sets   $b_1,b_2\in \Min{h}$.
  Then all four unions $a_i\cup b_j$ must be triangles (not just contain a triangle).
  Moreover, $a_1$ and $a_2$, as well as $b_1$ and $b_2$ must be incomparable under inclusion.

Case 1: Some of the sets  $a_1,a_2,b_1,b_2$  forms a triangle $T$, say $a_1=T$.
Hence, $b_1$ and $b_2$ lie in $T$, and $a_2\not\subseteq T$ since $a_1$ and $a_2$ must be
 incomparable. Consider the triangles $T_1=a_2\cup b_1$ and $T_2=a_2\cup b_2$.
If $|b_i|\geq 2$ for some $i\in\{1,2\}$, then $|T_i\cap T|\geq |b_i|\geq 2$, implying that
$T_i=T$, and hence, also $a_2\subseteq T$, a contradiction.
So, $b_1=\{e_1\}$ and $b_2=\{e_2\}$ for some edges $e_1\neq e_2$. Since then $|a_2|\geq 2$,
the triangles $T_1$ and $T_2$ are uniquely determined by $a_2$, implying that $T_1=T_2$
must be the same triangle. But this triangle shares two distinct edges $e_1$ and $e_2$ with $T$, implying that
$T_1=T$, and hence also  $a_2\subseteq T$, a contradiction.

Case 2: None of the sets  $a_1,a_2,b_1,b_2$ forms a triangle.
In this case, some of the sets must have exactly two edges, say $a_1=\{e_1,e_2\}$.
Since a triangle is uniquely determined by any two of its
edges, we have that both unions $a_1\cup b_1$ and $a_1\cup b_2$  must form the same triangle
$T=\{e_1,e_2,e_3\}$. The sets  $b_1$ and $b_2$ must be incomparable, and both of them must contain
the ``missing'' edge $e_3$. Since none of these two sets can be a triangle, this implies that
$b_1=\{e_1,e_3\}$ and $b_2=\{e_2,e_3\}$. These two sets also uniquely determine the
same triangle $T$, implying that $a_2\cup b_1=a_2\cup b_2=T$. Thus, $a_2$ must contain both
missing edges $e_1$ and $e_2$ of $T$. But this means that $a_2$ contains the set $a_1$,
 a contradiction with $a_1$ and $a_2$ being incomparable.
\end{proof}

We now turn to the proofs of our main results.

\section{Proof of Theorem~\ref{thm:main0}}

To show that the gap $\CC(\henv{f})/\CC(f)$ can be exponential,
consider the following polynomial in $n=m^2+m$ variables:
\begin{equation}\label{eq:pperm}
\PPPERM{n}(x,y)=\prod_{i=1}^m\sum_{j=1}^m x_{ij}y_j\,.
\end{equation}
The relation to the permanent polynomial $\PERM{}$ is that the coefficient of the monomial
$y_1y_2\cdots y_m$ in $\PPPERM{n}(x,y)$ is exactly $\PERM{m}(x)$.

Now, let $f(x,y)$ be the linearization of $\PPPERM{n}(x,y)$. That is, $f(x,y)$ is a
multilinear polynomial obtained from $\PPPERM{n}(x,y)$ by removing all nonzero exponents from
all monomials.
Every monomial of  $f$ has degree (sum of exponents) between $m+1$ and
$2m$, and the monomials
\[
x_{1,j_1}x_{2,j_2}\cdots x_{m,j_m} y_1y_2\cdots y_m
\]
 of degree $2m$
with all $j_1,\ldots, j_m$ distinct
are exactly the monomials of the polynomial
\[
h(x,y)=\PERM{m}(x)\cdot y_1y_2\cdots y_m\,.
\]
Thus, $h=\henv{f}$ is the higher envelope of $f$. Since $h(x,1,\ldots,1)=\PERM{m}(x)$,
Theorem~\ref{thm:perm} yields
\[
\C(\henv{f})\geq \C(\PERM{m})=2^{\Omega(m)}=2^{\Omega(\sqrt{n})}\,.
\]

On the other hand, since exponents play no role on $0$-$1$ inputs, we have that
$\PPPERM{}(a)=f(a)$ holds for all $0$-$1$ inputs $a$.
Thus, the polynomial $f$ can be counted by
the circuit given by the definition \eqref{eq:pperm} of $\PPPERM{}$. This gives the desired upper bound
$\C(f)=O(m^2)=O(n)$.

To show that the gap $\CC(\lenv{g})/\CC(g)$ can also be exponential,
consider the following polynomial in $n=m^2$ variables $x_{ij}$ given by the formula:
\begin{equation}\label{eq:isol}
\ISOL{n}(x)=\prod_{i=1}^m\prod_{j=m+1}^{2m}\bigg(\sum_{k=m+1}^{2m}x_{ik} \bigg)\bigg(\sum_{l=1}^{m}x_{lj} \bigg)\,.
\end{equation}
The monomials of this polynomial are obtained as follows.
We interpret the variables $x_{ij}$ as edges of a complete bipartite $m\times m$ graph $I\times J$
with parts $I=\{1,\ldots,m\}$ and $J=\{m+1,\ldots,2m\}$.
To get a monomial of $\ISOL{}$, we take, for each node $i\in I$ exactly one
edge $x_{ik}$ incident with $i$, and take, for each node $j\in J$ exactly one
edge $x_{lj}$ incident with $j$. So, every variable has degree at
most $2$. Note that on every $0$-$1$ input $a\in\{0,1\}^n$,
$\ISOL{}(a)=0$ if and only if the graph specified by $a$ has an isolated node.

Let $g$ be the linearization of $\ISOL{n}$.
Every monomial of  $g$ has degree between $m$ and
$2m$, and the monomials of degree $m$ correspond to perfect
matchings. Thus, the lower envelope $\lenv{g}$ of $g$ is just the permanent polynomial, i.e. $\lenv{g}=\PERM{m}$. By Theorem~\ref{thm:perm},
$\CC(\lenv{g})=2^{\Omega(m)}$.

On the other hand, since exponents play no role on $0$-$1$ inputs, we have that
$\ISOL{}(a)=g(a)$ holds for all $0$-$1$ inputs $a$.
Thus, the polynomial $g$ can be counted by
the circuit given by the definition \eqref{eq:isol} of $\ISOL{}$. This gives the desired upper bound
$\C(g)=O(m^3)=O(n^{3/2})$. \qed

\section{Proof of Theorem~\ref{thm:main02}}
\label{sec:reduction}

Recall that the $s$-$t$ path polynomial $f=\PATH{n}$ has one variable $x_{i,j}$ for each edge of a complete undirected graph on $n+2$ nodes $\{s,1,\ldots,n,t\}$.  Each monomial of $f$ corresponds to a simple directed path from node $s$ to node $t$.

The upper bound $\dec(f)=O(n^3)$ of the \emph{decision} complexity of $f=\PATH{n}$
follows from the Bellman--Ford dynamic programming algorithm \cite{bellman,ford}.
  The circuit is constructed recursively by taking
$F_{1,j}=x_{s,j}$ for all $j\in[n]\cup\{t\}$, and using the recursion $F_{l+1,j}=F_{l,j}+\sum_{i\neq j}F_{l,i}\prd x_{i,j}$ for $j\in[n]\cup\{t\}$. Monomials of $F_{l,j}$
correspond to walks from node $s$ to node $j$ passing through at most $l$ edges; one edge may be passed more than once, and each pass counts. The output is the polynomial $F=F_{n+1,t}$.
Since every $s$-$t$ walk contains a simple $s$-$t$ path, and since in deciding $\{+,\prd\}$-circuits we can use the absorption axiom $x+xy=x$, the circuit correctly \emph{decides} $\PATH{n}$. Thus $\dec(\PATH{n})=O(n^3)$.

Our goal is now to show that every $\{+,\prd\}$-circuit \emph{counting} the
$s$-$t$ path polynomial must have exponential size.
We do not have a \emph{direct} proof of this lower bound. Instead, we will
derive this result indirectly by using some known reductions and lower bounds.

Say that a $\{+,\prd\}$-circuit \emph{decides} $f$ \emph{with threshold $T$},
if for every $a\in\{0,1\}^n$,
$F(a)\geq T$ holds precisely when $f(a)\geq 1$. Here, the threshold $T=T(n)$ may depend
on the number $n$ of variables, but not on the input. Note that deciding $\{+,\prd\}$-circuits decide with threshold $T=1$.
Let $\dect(f)$ denote the smallest size of a  $\{+,\prd\}$-circuit deciding~$f$ with
some threshold~$T$.

As defined by Valiant~\cite{valiant79}, and Skyum and Valiant~\cite{SkyumV85},
a polynomial $f(x_1,\ldots,x_n)$ is a \emph{monotone projection} of a polynomial $g(y_1,\ldots,y_m)$  if there exists an assignment $\sigma:\{y_1,\ldots,y_m\}\to\{x_1,\ldots,x_n, 0,1\}$ such that $f(x_1,\ldots,x_n)=g(\sigma(y_1),\ldots,\sigma(y_m))$.
It is clear that then $\dect(f)\leq \dect(g)$.

The \emph{$r$-clique polynomial}, $\Clique{n,r}$, has  $\tbinom{n}{2}$
variables $x_e$, one for each edge $e$ of $K_n$, and has
one monomial $\prod_{e\subseteq S}x_e$ for every subset $S\subseteq [n]$ of size $|S|=r$.
Results of Valiant \cite{valiant79a} imply that, for every $1\leq r\leq n$,  $\Clique{n,r}$ is a monotone projection
of  the Hamiltonian $s$-$t$ path polynomial $\HP{m}$ for $m=n^{O(1)}$; as noted by Alon and Boppana~\cite{AB87}, already $m=25n^2$
is enough in this case.
On the other hand, it is known that, for $r$ about $\sqrt{n}$, the clique polynomial $f=\Clique{n,r}$
requires  $\dect(f)\geq 2^{n^{\Omega(1)}}$ \cite{haken99,pudlak,Juk99}; see, e.g. \cite[Sect. 9.8]{myBFC-book} for a simpler proof. (In fact, this result holds
 for more general circuits where arbitrary monotone real valued functions
  $g:\RR^2\to\RR$ can be used as gates.) Since $\Clique{n,r}$ is a monotone projection of $\HP{m}$, we have that
\[
\dect(\HP{m})\geq \dect(\Clique{n,r})=2^{n^{\Omega(1)}}\,.
\]
It remains therefore to show that
\[
\mbox{$\C(\PATH{m})\geq \dect(\HP{n})$ for $m=n^{O(1)}$.}
\]
This can be shown using a standard reduction of $\PATH{}$ to $\HP{}$.
Let $p=(n+1)\log n$.
Given an input graph $G$ on $n+2$ nodes $\{s,1,2,\ldots,n,t\}$, replace each edge $(u,v)$ by a graph on $2p+2$ nodes ($u$, $v$ and $2p$ new nodes) containing exactly $2^p$
paths of length $p+1$ between $u$ and $v$.
This way, every $s$-$t$ path of length $l$ in $G$ gives $(2^p)^l$ $s$-$t$ paths in the resulting graph  $G'$. This graph has $m=O(pn^2)=O(n^3\log n)$ nodes.

If $G$ has a Hamiltonian $s$-$t$ path (of length $n+1$), then the graph $G'$ has at least
$T:=(2^p)^{n+1}$ $s$-$t$ paths. If $G$ has no Hamiltonian path, then the longest
$s$-$t$ path has at most $n$ edges, and hence, at most $n-1$ inner nodes. The number of $s$-$t$ paths of length $\leq n$ is bounded from above by $n\cdot n^{n-1}=n^{n}$. So, in this case, $G'$ can have at most
$(2^p)^{n}\cdot n^{n}=l\cdot n^n/2^p=T/n$ $s$-$t$ paths.
 We have thus shown that every $\{+,\prd\}$-circuit counting $\PATH{m}$ for $m=\Theta(pn^2)=\Theta(n^3\log n)$ decides  $\HP{n}$ with threshold $T=(2^p)^{n+1}$.
\qed

\section{Proof of Lemma~\ref{lem:monic}}

Let $f(x_1,\ldots,x_n)$ be a polynomial in which each variable $x_i$ has degree at most $t_i$,
and let $S_i\subseteq \NN$ be arbitrary subsets of sizes $|S_i|\geq t_i+1$, $i=1,\ldots,n$.

\begin{clm}[Folklore]\label{clm:uniq}
 The polynomial $f$ is uniquely determined by its values on
 $S_1\times S_2\times \cdots\times S_n$.
\end{clm}

\begin{proof}
Induction on $n$. For $n=1$, the claim is simply the assertion that a non-zero
polynomial of degree $t_1$ in one variable can have at most $t_1$ distinct roots.
For the induction step, expand the polynomial $f$ by the variable $x_n$:
\[
f(x_1,\ldots,x_n)=\sum_{i=0}^{t_n}f_i(x_1,\ldots,x_{n-1})\cdot x_n^i\,.
\]
For each point $a\in S_1\times\cdots\times S_{n-1}$,
 \[
f(a,x_n)=\sum_{i=0}^{t_n} f_i(a)\cdot x_n^i
\]
is a polynomial of degree at most $t_n$ in one variable, and hence, all its coefficients $f_i(a)$, $i=0,1,\ldots,t_n$ can be recovered knowing the values $f(a,b)$ for all $b\in S_{n+1}$. Knowing the values $f_i(a)$ for all $a\in S_1\times\cdots\times S_{n-1}$ we can, by the
induction hypothesis, recover the polynomials $f_i$, and hence, the original polynomial~$f$.
\end{proof}

Now let $f$ and $h$ be two polynomials on the same set of $n$ variables such that
$f(a)=h(a)$, and hence, also $\lin{f}(a)=\lin{h}(a)$  holds for all $a\in\{0,1\}^n$.
(Recall that $\lin{f}$ is obtained from $f$ by removing all nonzero exponents.)
Since the polynomials $\lin{f}$ and $\lin{h}$ are multilinear, Claim~\ref{clm:uniq} with all
$S_i=\{0,1\}$ yields
$\lin{f}=\lin{h}$ (they must coincide as multilinear polynomials), and hence, also
$\Sup{f}=\Sup{h}$ must hold as well.

Let us now prove the second claim of Lemma~\ref{lem:monic}: if $f$ and $h$ are polynomials
on the same set of variables, then $f$ and $h$ have the same $0$-$1$ roots
if and only if  $\Min{f}=\Min{h}$. The ``if'' part is trivial, because $f(a)>0$ happens precisely when $p(a)=1$ for some monomial $p\in\Min{f}$. To prove the ``only if'' direction, assume that $f$ and $h$ have the same $0$-$1$ roots.
 Our goal is to show that then
 $\Min{f}=\Min{h}$ must hold.

 Assume contrariwise that there is a monomial $p\in \mon{f}$ whose set of variables $X_p$ belongs to $\Min{f}$ but not to $\Min{h}$. If $\msup{q}\not\subseteq \msup{p}$ holds for
 all monomials $q$ of $h$, then we can set all variables in $\msup{p}$ to $1$ and the rest
 to $0$. On the resulting assignment $a=a_p$, we will have $h(a)=0$ but $f(a)\geq p(a)\geq 1$,
 a contradiction.
 Thus, there must be a monomial $q\in\Min{h}$ such that $\msup{q}\subset \msup{p}$; the inclusion must be proper, because $\msup{p}\not\in\Min{h}$. But then on the input $a_q$, we will have
 $f(a_q)=0$ but $h(a_q)\geq q(a_q)\geq 1$, a contradiction again.
\qed

\section{Proof of Lemma~\ref{lem:balan}}

We will need the following two simple and well-known facts.

A \emph{subadditive weighting} of a circuit is an assignment of nonnegative numbers (weights)
to its gates such that the weight of a gate does
not exceed the sum of the weights of its inputs.

\begin{clm}[Folklore]\label{clm:spira}
If the output gate gets weight $m$, and every leaf gets weight at most $2m/3$, then
there is a gate of weight larger than $m/3$ and at most $2m/3$.
\end{clm}

\begin{proof}
 By starting at the output gate, and traversing the circuit
 by always choosing the input of larger weight, we can find a gate $v$ of
 weight $> 2m/3$ such
  that both its inputs $u$ and $w$ have weights at most $2m/3$.
  By the subadditivity of weighting, at least one of the gates $u$ and $w$
  have then weight larger than $(2m/3)/2=m/3$ and at most $2m/3$.
\end{proof}

\begin{clm}[Folklore]\label{clm:contain}
For every gate $u$ in a $\{+,\prd\}$-circuit producing a polynomial $F$, the polynomial
can be written as $F=P\pr Q+R$, where $P$ is the polynomial produced at~$u$.
\end{clm}

(We use capital letters for polynomials only to stress that they are produced by circuits.)

\begin{proof}
  If we replace the gate $u$ by a new variable $y$, the resulting
  circuit produces a polynomial of the form $y\pr H+R$
  for some polynomial $H$, where $R$ does not contain $y$ (albeit $H$ may contain).  It remains to substitute all occurrences
  of the variable $y$ with the polynomial $P$ produced at
  the gate~$u$.
\end{proof}

\begin{proof}[Proof of Lemma~\ref{lem:balan}(i)]
For a polynomial $f$, let $\length{f}$ denote the minimum number of variables in a monomial of $f$. Hence, a product $g\pr h$ of two polynomials is $m$-balanced, if  $m/3<\length{g}\leq 2m/3$. We have to show that, if $\length{f}\geq m$ for $m\geq 2$, then
$\Sup{f}$ is a union of at most $s=\C(f)$ supports of $m$-balanced products of polynomials.

To prove this claim, fix a $\{+,\prd\}$-circuit of size $s=\C(f)$ counting $f$.
  Define the \emph{weight} of a gate $u$ as $\length{P}$, where $P=P_u$
  is the polynomial produced at $u$. Hence, the output gate has weight
  at least $m\geq 2$, and each input gate has weight $1$ (which is $\leq 2m/3$
   since $m\geq 2$). Since this weighting is subadditive,
  Claim~\ref{clm:spira} gives us a gate $u$ with $m/3 < \length{P}\leq 2m/3$.
By Claim~\ref{clm:contain}, we can write the produced by our circuit polynomial $F$ as
a sum $F=P\pr Q + R$. Hence, $\Sup{f}=\Sup{F}=\Sup{P\pr Q}\cup \Sup{R}$, where the product $P\pr Q$ is $m$-balanced.

The polynomial
$R$ is obtained from $F$ by removing some monomials. If
$R$ is empty, then we are done.  Otherwise,
the polynomial $R$ can be produced by a circuit with one gate fewer (gate $u$ is set to constant $0$, and disappears).
Moreover, $\mon{R}\subseteq \mon{F}$ implies that $\length{R}\geq \length{F}\geq m$
still holds. So, we can
repeat the same argument for the polynomial $R$, until
the empty polynomial $R$ is obtained.
\end{proof}

\begin{proof}[Proof of Lemma~\ref{lem:balan}(ii)]
We will now apply Claim~\ref{clm:spira} not to the entire
circuit but to some its sub-circuits.
A \emph{parse-subcircuit} of a circuit $\F$ is obtained by setting to $0$ one of the two
inputs of each sum gate. Such a subcircuit $\F'$ can also be defined inductively as follows.
The output gate of $\F$ is included in $\F'$. If a gate $u$ is already included in $\F'$,
and if $u$ is a sum gate, then exactly one of the inputs to $u$ are included in $\F'$. If
$u$ is a product gate, then both its inputs are included in~$\F'$
(see Fig.~\ref{fig:circuit1}). Note that each parse-subcircuit produces
exactly one monomial in a natural way, and that each monomial of the polynomial
produced by the entire circuit
is produced by at least one parse-subcircuit.

\begin{figure}
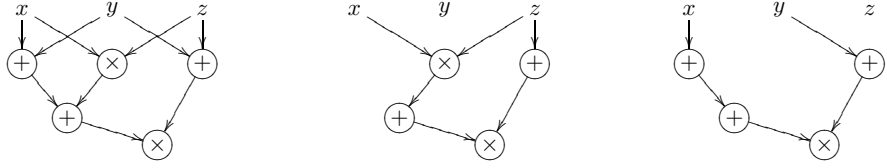

 \[
\scalebox{0.8}{
 \xygraph { !{<0cm,0cm>;<1cm,0cm>:<0cm,-0.9cm>::}
    !{(0,0)}*+{x}="x"
    !{(1.5,0)}*+{y}="y"
    !{(3,0)}*+{z}="z"
    !{(0,1)}*+=[o]+[F]{+}="s1"
    !{(1.5,1)}*+=[o]+[F]{\prd}="s2"
    !{(3,1)}*+=[o]+[F]{+}="s3"
    !{(0.75,2)}*+=[o]+[F]{+}="p1"
    !{(2.25,2.5)}*+=[o]+[F]{\prd}="p2"
    "x":"s1" "x":"s2" "y":"s1" "y":"s3" "z":"s2" "z":"s3" "s1":"p1"
    "s2":"p1" "s3":"p2" "p1":"p2"
    }
    }
    \qquad\qquad
 \scalebox{0.8}{
 \xygraph { !{<0cm,0cm>;<1cm,0cm>:<0cm,-0.9cm>::}
    !{(0,0)}*+{x}="x"
    !{(1.5,0)}*+{y}="y"
    !{(3,0)}*+{z}="z"
    !{(1.5,1)}*+=[o]+[F]{\prd}="s2"
    !{(3,1)}*+=[o]+[F]{+}="s3"
    !{(0.75,2)}*+=[o]+[F]{+}="p1"
    !{(2.25,2.5)}*+=[o]+[F]{\prd}="p2"
     "x":"s2"   "z":"s2"  "z":"s3"
    "s2":"p1" "s3":"p2" "p1":"p2"
    }
    }
    \qquad \qquad
    \scalebox{0.8}{
 \xygraph { !{<0cm,0cm>;<1cm,0cm>:<0cm,-0.9cm>::}
    !{(0,0)}*+{x}="x"
    !{(1.5,0)}*+{y}="y"
    !{(3,0)}*+{z}="z"
    !{(0,1)}*+=[o]+[F]{+}="s1"
    !{(3,1)}*+=[o]+[F]{+}="s3"
    !{(0.75,2)}*+=[o]+[F]{+}="p1"
    !{(2.25,2.5)}*+=[o]+[F]{\prd}="p2"
     "y":"s3" "x":"s1"
    "s1":"p1" "s3":"p2" "p1":"p2"
    }
    }
    \]
    \caption{A circuit and two its parse sub-circuits producing, respectively,
    the monomials $xz^2$ and $xy$.
    }
  \label{fig:circuit1}
\end{figure}

Now let $\F$ be a circuit of size $s=\C(f)$ counting
$f$, and $F$ be the polynomial produced by $\F$. By Lemma~\ref{lem:monic}, we have that $\Sup{f}=\Sup{F}$.
For every monomial $p$ of $F$ of length at least $m$, take some parse-subcircuit $\F_p$ producing $p$, and use
Claim~\ref{clm:spira} to find a gate $u$ in $\F_p$ such that the part $p'$ of $p$
produced at $u$ in $\F_p$ has length $l$ satisfying $m/3<l\leq 2m/3$.
By Claim~\ref{clm:contain}, we can write the polynomial $F$ as
a sum $F=P\pr Q + R$, where $P$ is the polynomial produced at gate $u$ (in the entire circuit).
Hence, $p$ appears $m$-balanced in the
product $R_u=P\pr Q$.
 Since we have at most $s$ products $R_u$, and since
$\mon{R_u}\subseteq \mon{F}$ implies  $\Sup{R_u}\subseteq \Sup{f}$,
we are done.
\end{proof}

\section{Proof of Theorems~\ref{thm:main} and \ref{thm:main01}}

Define the \emph{join} of two families of sets $B$
and $C$ as the family
\[
B\joi C=\{b\cup c\colon b\in B, c\in C\}
\]
of all possible unions. Note that
the support of a product $g\pr h$ of two polynomials  is the join of the
supports of $g$ and $h$. Note also that, if no set of $B$ intersects any set of $C$, then
we have an upper bound
$|A|\leq \ddeg{A}{|b|}\cdot \ddeg{A}{|c|}$ on the size of the join $A=B\joi C$
holding for \emph{all} $b\in B$ and $c\in C$. This holds because then
$|B|=|B\joi \{c\}|\leq \ddeg{A}{|c|}$, and similarly $|C|=|\{b\}\joi C|\leq
\ddeg{A}{|b|}$.
If, however, sets in $B$ and in $C$ intersect, then it may happen that
$|B|\gg  |B\joi \{c\}|$.
Still, also then we have a reasonable upper bound.

\begin{lem}\label{lem:sharp}
  Let $B\joi C$ be a join of two families, and $B\joi C\subseteq A$. Suppose that every set in $B\joi C$ has size at least $m$, and that $B$ or $C$ has a set of size $r$. Then
\[
|B\joi C|\leq \ddeg{A}{r}\cdot \ddeg{A}{m-r}\,.
\]
\end{lem}

\begin{proof} Assume w.l.o.g. that the family $B$ contains a set $b$ of size $|b|=r$,
and let $A_b=\{b\}\joi C\subseteq A$. Associate with every
  $a\in A_b$ the family
  \[
  C_a=\{c\in C\colon b\cup c=a\}\,.
  \]
  These families
  give a partition of $C$ into $|A_b|$ pairwise disjoint
  subfamilies. Since all sets in $A_b$ contain the set $b$ of size $|b|=r$, we have that
  \[
  |A_b|\leq \ddeg{A}{r}\,.
  \]
   On the other hand, for each $a\in A_b$,
  all sets in $C_a$, and hence, also all sets in $B\joi C_a$ contain
  the set $a\setminus b$ of size $|a\setminus b|\geq m-r$, implying that
  \[
  |B\joi
  C_a|\leq\ddeg{A}{m-r}
  \]
   holds for all
  $a\in A_b$. Now, every set $b'\cup c'$ in $B\joi C$ belongs to $B\joi C_a$ for
  $a=b\cup c'$. So,
  \begin{align*}
  |B\joi C|&\leq \sum_{a\in A_b}|B\joi C_a|\leq \sum_{a\in A_b}\ddeg{A}{m-r}\\
    &\leq |A_b|\cdot \ddeg{A}{m-r}\leq \ddeg{A}{r}\cdot
    \ddeg{A}{m-r}\,. \qedhere
  \end{align*}
\end{proof}

\begin{proof}[Proof of Theorem~\ref{thm:main}]
Let $f=g+h$ be a polynomial such that $\length{g}\geq m\geq 2$, and $\length{h}<m/3$; here, as before, $\length{f}$ denotes the minimum number of variables in a monomial of $f$.
By Lemma~\ref{lem:balan}(ii), there are  $s=\C(f)$  products $P\pr Q$ of polynomials
such that $\Sup{P\pr Q}\subseteq \Sup{f}$, and
every monomial of $g$
appears $m$-balanced in at least one of these products.

\begin{clm}\label{clm:bal}
If $\Sup{P\pr Q}\subseteq \Sup{f}$, and if at least one monomial of $g$
appears $m$-balanced in $P\pr Q$, then $\Sup{P\pr Q}\subseteq \Sup{g}$
and $|\Sup{P\pr Q}|\leq \ddeg{g}{r}\cdot \ddeg{g}{m-r}$ for some
$m/3<r\leq 2m/3$.
\end{clm}

\begin{proof}
To show the inclusion $\Sup{P\pr Q}\subseteq \Sup{g}$,
assume contrariwise that there are $a,a'\in \Sup{P}$ and
$b,b'\in \Sup{Q}$ such that $a\cup b\in \Sup{g}$, $m/3 < |a|\leq 2m/3$ but $a'\cup b'\in\Sup{h}$.  Since $|b'|=l$ for some
$l<m/3$, the union $a\cup b'$ has size
$l< m/3 <  |a\cup b'|\leq 2m/3+l <m$,
and hence, cannot belong to $\Sup{f}$, a contradiction with $\Sup{P\pr Q}\subseteq \Sup{f}$. Thus, $\Sup{P\pr Q}$ must lie entirely within $\Sup{g}$, as claimed.

To show the
upper bound on $|\Sup{P\pr Q}|$, let $A=\Sup{g}$, $B=\Sup{P}$ and $C=\Sup{Q}$.
Since $\length{g}\geq m$, and $\Sup{P\pr Q}\subseteq \Sup{g}$, we have that
every set in $B\joi C=\Sup{P\pr Q}$ has at least $m$ elements.
On the other hand, since some monomial of $g$
appears $m$-balanced in $P\pr Q$, some set in $B$ must have $r$ elements, for some
$m/3<r\leq 2m/3$. For this $r$, Lemma~\ref{lem:sharp} yields
$|A\joi B|=|\Sup{P\pr Q}|\leq \ddeg{A}{r}\cdot \ddeg{A}{m-r}$, as desired.
\end{proof}

Thus, every monomial of $g$ belongs to at least one of $s$ products
$P\pr Q$ of polynomials such that $|\Sup{P\pr Q}|\leq \ddeg{g}{r}\cdot \ddeg{g}{m-r}$
for some $m/3<r\leq 2m/r$. By taking such an $r$ maximizing $\ddeg{g}{r}\cdot \ddeg{g}{m-r}$,
the desired lower bound $s\geq |\Sup{g}|/\ddeg{g}{r}\cdot \ddeg{g}{m-r}$
follows.
\end{proof}

\begin{proof}[Proof of Theorem~\ref{thm:main01}]
Recall that our polynomial $f$ has the form $f=g+h$ with $g=\PERM{n}$ and
$h=\sum_{i,j\in[n]}x_{ij}$. Hence, $\length{g}=n$ and $\length{h}=1 < n/3$.
By Theorem~\ref{thm:main}, there is an integer $r$ between $n/3$
 and $2n/3$ such that $\C(f)\geq |\Sup{g}|/\ddeg{g}{r}\cdot\ddeg{g}{m-r}
 \geq n!/r!(n-r)!=2^{\Omega(n)}$.
On the other hand, on every $0$-$1$ input $a$, we have that $f(a)=0$ if and only if $h(a)=0$, because $g(0,\ldots,0)=0$. Hence, the circuit $h$ decides~$f$,
implying that $\dec(f)=\dec(h)\leq n^2$.
\end{proof}

\section{Proof of Theorem~\ref{thm:main2}}

By Claim~\ref{clm:contain}, we know that, for every gate $u$ in a given $\{+,\prd\}$-circuit $\F$, the produced by the circuit polynomial $F$
can be written as $F=P_{u}\pr Q_{u}+R$,
where $P_{u}$ is the polynomial produced at~$u$,
$Q_{u}$ is the polynomial produced ``after'' the \emph{gate} $u$, and
$R$ is the polynomial produced by the circuit after the gate $u$ is replaced with constant~$0$.
For our argument, it will be convenient to introduce the notion of a  polynomial
$Q_{e}$ produced after an \emph{edge} $e=(u,v)$ (see Fig.~\ref{fig:circuit2}):
 \[
 Q_{e}=\begin{cases}
 Q_{v} & \mbox { if $v=u+w$,}\\
 Q_{v}\pr P_{w} & \mbox { if $v=u\prd w$.}
 \end{cases}
 \]
\begin{SCfigure}[10]
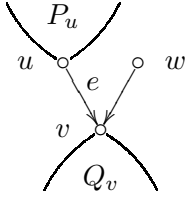

 \xygraph { !{<0cm,0cm>;<1cm,0cm>:<0cm,-0.9cm>::}
    !{(0,0)}*+=[o]{\circ}="u"
    !{(-0.5,0)}*+{u}="uu"
    !{(0,-0.7)}*+{P_u}="P"
    !{(0.4,0.3)}*+{e}="e"
    !{(0.5,1)}*+=[o]{\circ}="v"
    !{(0,1)}*+{v}="vv"
    !{(0.5,1.7)}*+{Q_v}="Q"
    !{(1,0)}*+=[o]{\circ}="w"
    !{(1.5,0)}*+{w}="ww"
    !{(-0.8,-1)}*+{}="a"
    !{(0.8,-1)}*+{}="b"
    !{(-0.3,2)}*+{}="c"
    !{(1.3,2)}*+{}="d"
    "a"-@/_0.1cm/"u"
    "b"-@/^0.1cm/"u"
    "u":"v" "w":"v"
    "v"-@/_0.1cm/"c"
    "v"-@/^0.1cm/"d"
    }
    \caption{For an edge $e=(u,v)$, the polynomial $Q_e$ produced after $e$ is
    the polynomial $Q_e=Q_v$ produced after the gate $v$, if $v=u+w$ is a sum gate,
    and is $Q_e=Q_v\pr P_w$, if $v=u\prd w$ is a product gate, where $P_w$ is the polynomial produced before the gate $w$.}
  \label{fig:circuit2}
\end{SCfigure}
A set $E$ of edges of $\F$ is a \emph{cut},
if every   input-output path in $\F$ contains an edge in $E$.

\begin{clm}\label{clm:cuts}
  If $E$ is a cut, then $\mon{F}$ is a union of $\mon{P_{u}\pr Q_{e}}$ over all edges $e=(u,v)$
  in~$E$.
\end{clm}

\begin{proof}
Take a monomial $p$ of the produced polynomial $F$,
and let $\F_p$ be any parse-subcircuit producing $p$. Since $E$ forms a
cut, the graph $\F_p$ must contain some edge $e=(u,v)\in E$.
Then the monomial $p$ has the form $p=p'p''$ where $p'$ is the monomial
produced by the subgraph of $\F_p$ rooted in $u$. Thus $p'$ belongs to the polynomial $P_{u}$
produced in $\F$ before the edge $e$, and  $p''$ belongs to the polynomial $Q_{e}$
produced after the edge~$e$.
Hence, $p$ belongs to $P_{u}\pr Q_{e}$, as desired.
\end{proof}

\begin{proof}[Proof of Theorem~\ref{thm:main2}]
Let $\F$ be a $\{+,\prd\}$-circuit of size $s=\C(f)$ counting $f$, and let $F$ be the polynomial produced by $\F$. By Lemma~\ref{lem:monic}, we know that
 $\Sup{F}=\Sup{f}$. Hence, the polynomial $F$ is also $(k,l)$-free.
We first transform the circuit $\F$ to a circuit $\F'$ as follows.
For every product gate $v=u\prd w$ in $\F$, one of whose inputs, say $u$, is \emph{small} in that  $\m{P_{u}}\leq l$ holds, we remove the edge $(u,v)$ and replace $v$ by
a unary (fanin-$1$) gate $v=P_{u}\prd w$ of ``scalar'' multiplication by this
  fixed (small) polynomial $P_{u}$.  If both inputs produce small
  polynomials, then we eliminate only one of them.
  It is clear that $\F'$ produces the same polynomial $F$.
  In particular,  $\Sup{F'}=\Sup{f}$ holds as well.

Say that  an edge $e=(u,v)$ of $\F'$ is \emph{light}, if $\m{P_{u}\pr Q_{e}}\leq kl^2$.
To finish the proof of the first claim in Theorem~\ref{thm:main2}, it is enough, by Claim~\ref{clm:cuts}, to show that
every input-output path in $\F'$ must contain at least one light edge.

To show this, take an arbitrary input-output path  in $\F'$, and
let $e=(u,v)$ be the last edge along this path such that $\m{P_{u}}\leq k$; hence,
$\m{P_{v}}> k$.
Such an edge must exist because $\m{x_i}=1\leq k$, and since we can assume that $\m{F}>k$
(for otherwise the theorem would trivially hold).
Together with $\Min{P_{v}\pr Q_{v}}\subseteq \Min{F}$ and  $\m{P_{v}}> k$, the $(k,l)$-freeness of $F$ implies that
\[
\m{Q_{v}}\leq l\,.
\]
If $v$ is a sum gate, then $Q_{e}=Q_{v}$, and hence, also $\m{Q_{e}}\leq l$.
So, the edge $e$ is light in this case:
\[
\m{P_{u}\pr Q_{e}}\leq \m{P_{u}}\cdot\m{Q_{e}}\leq kl\,.
\]
So, assume that $v$ is a product gate. Let $u$ and $w$ be the inputs to $v$ in the original
circuit $\F$. Since $\m{P_{u}}\leq k\leq l$, we have that $\m{P_{w}}\leq l$ must hold as well, for otherwise the edge $e=(u,v)$ could not exist in $\F'$ (would be already eliminated when going from $\F$ to $\F'$). Hence,
\[
\m{Q_{e}}=\m{P_{w}\pr Q_{v}}\leq l^2\,.
\]
So, the edge $e$ is light also in this case:
\[
\m{P_{u}\pr Q_{e}}\leq \m{P_{u}}\cdot\m{Q_{e}}\leq kl^2\,.
\]
Since the total number of edges in $\F'$ is at most $2s$,
we have thus shown that the support $\Sup{F'}=\Sup{f}$ is a union of
at most $2s$ families $\Sup{P\pr Q}$ with $\m{P\pr Q}\leq kl^2$.
Since every minimal set of a union of two families must be minimal in at least one of these families, this implies that $\Min{f}$ is contained in (albeit not necessarily equal to) the union of the families $\Min{P\pr Q}$. Hence, the desired lower bound $s\geq |\Min{f}|/2lk^2$.
\end{proof}

\section{Conclusion and Open Problems}

The weakness of monotone arithmetic circuits, i.e. of $\{+,\prd\}$-circuits, \emph{computing}  a given polynomial $f$ is stipulated by the fact that the produced by the circuit polynomial $F$ must just (syntactically) coincide with $f$. In particular, then $\mon{F}=\mon{f}$ must hold. On the other pole are $\{+,\prd\}$-circuits just \emph{deciding} $f$. These are, in fact, monotone boolean circuits, where the idempotence axiom $x^2=x$ as well as the
absorption axiom $x+xy=x$ can be used, and hence, here we only have a weaker property
$\Min{F}=\Min{f}$. While proving lower bounds in the latter (boolean) model is a relatively difficult task, the severe restriction of the former (arithmetic) model makes this task much easier.

In this paper we considered an intermediate model of \emph{counting}  $\{+,\prd\}$-circuits. In this case, it is required that the values of $F$
must coincide with those of $f$ on only $0$-$1$ inputs: on other inputs, the values may be different. Thus, counting circuits are $\{+,\prd\}$-circuits that are allowed to use the
idempotence axiom $x^2=x$ (but not the
absorption axiom $x+xy=x$). These circuits have an intermediate
structural property that $\Sup{F}=\Sup{f}$ must hold (Lemma~\ref{lem:monic}).
We have shown that counting circuits can be exponentially smaller than computing circuits
(Theorem~\ref{thm:main0}), and that  deciding circuits can be exponentially smaller than
counting circuits (Theorem~\ref{thm:main01}).

\begin{table}
  \begin{center}
    \begin{tabular*}{1\textwidth}{l@{\hskip 0.5cm}c@{\hskip 0.5cm}c@{\hskip 0.5cm}c@{\hskip 0.5cm}c}
       \hline\noalign{\smallskip}
      Circuits &  $x+x=x$ & $x^2=x$ & $x+xy=x$ & Property\\
      \noalign{\smallskip}\hline\noalign{\smallskip}
      Computing  & $-$ & $-$ & $-$ & $F=f$\\[0.5ex]
      Counting & $-$ & $\checkmark$ & $-$ & $\lin{F}=\lin{f}$\\[0.5ex]
      Approximating & $\checkmark$ & $\checkmark$& $-$ & $\Sup{F}=\Sup{f}$\\[0.5ex]
      Tropical & $\checkmark$ & $-$ & $\checkmark$ & $\MMin{F}=\MMin{f}$\\[0.5ex]
      Deciding/Boolean & $\checkmark$ & $\checkmark$ & $\checkmark$ & $\Min{F}=\Min{f}$\\
      \noalign{\smallskip}\hline
    \end{tabular*}
    \caption[]{Summary of which axioms are allowed $(\checkmark)$ in which kind of $\{+,\prd\}$-circuits.  The last column indicates what property the produced by a circuit polynomial $F$ must satisfy; here $\lin{f}$ is the
    linearization of $f$ obtained by removing all nonzero exponents.
    Tropical circuits are circuits with
    $x\oplus y=\min(x,y)$ and $x\otimes y=x+y$ functions as gates.
      Finally, $\MMin{f}$ is the set of all \emph{monomials} of $f$ that contain no other monomial of $f$ as a proper factor.
      The property $\MMin{F}=\MMin{f}$ holds only if $f$ is multilinear~\cite{jerrum,juk14}.}
    \label{tab:axioms}
  \end{center}
\end{table}

A next natural question was whether lower-bounds arguments
for the weak (computing) model can be extended to work also for the intermediate (counting) model? We have shown that such an extension is possible for two lower-bounds
arguments (Theorems~\ref{thm:main}--\ref{thm:main2}).
In fact, our proofs of these bounds hold for  $\{+,\prd\}$-circuits that only ``approximate'' a given polynomial $f$ in that
$\Sup{F}=\Sup{f}$ holds for the produced by the circuit polynomial~$F$ (coefficients play no role in our arguments).
Approximating circuits can use both idempotence axioms $x+x=x$ and $x^2=x$.
(Table~\ref{tab:axioms} summarizes the axioms allowed in various types of circuits.)
So, these bounds also hold for $\{\cup,\joi\}$-circuits constructing a given family $A\subseteq 2^X$ of subsets
of a (fixed) finite set. Inputs are single element sets $\{x\}$ with $x\in X$, and gates are set-theoretic union $(\cup)$ and join $(\joi)$ of families.
A special case of Theorem~\ref{thm:main} (for $h=0$) gives that, if every set in $A$ has at least $m\geq 2$ elements, then there is an integer $m/3<r\leq 2m/3$ such that every
$\{\cup,\joi\}$-circuit constructing $A$ must have at least
$|A|/\ddeg{A}{r}\cdot\ddeg{A}{m-r}$ gates.

A ``complementary'' in a sense to counting $\{+,\prd\}$-circuits model, also
lying between computing $\{+,\prd\}$-circuits and deciding $\{+,\prd\}$-circuits,
is that of \emph{tropical} circuits, i.e. $\{\min,+\}$-circuits.  These are $\{+,\prd\}$-circuits, where the sum is interpreted as $\min\{x,y\}$, and the product as $x+y$. Such a circuit \emph{computes} a given
polynomial $f$ of $n$ variables, if $\trop{F}(a)=\trop{f}(a)$ holds for all $a\in\NN^n$, where $\trop{f}$ is the ``tropicalization'' of~$f$:
\[
f(x)=\sum_{e\in \NN^n}c_e\prod_{i=1}^n x_i^{e_i}\qquad \mbox{ turns to }\qquad
\trop{f}(x)=\min_{\substack{e\in \NN^n\\ c_e\neq 0}}\ \sum_{i=1}^n e_ix_i\,.
\]
For example, if
$f=xy^2+ 3y^2z^3$, then $\trop{f}=\min\{x+2y, 2y+3z\}$.
Tropical circuits are important, because many dynamic programming algorithms for minimization
problems are just recursively constructed tropical circuits.

The difference from counting $\{+,\prd\}$-circuits is that now the absorption axiom $x+xy=x$ is allowed, but the idempotence axiom $x^2=x$ is not ($x+x\neq x$ unless $x=0$).
As shown in \cite{jerrum,juk14}, lower bounds for computing $\{+,\prd\}$-circuits
hold also for tropical circuits, as long as the target polynomial $f$ is multilinear:
in this case we have that $\T(f)\geq \CC(\lenv{f})$, where $\T(f)$ is the minimum size of
a tropical circuit computing~$f$. In particular, for polynomials which are multilinear
and homogeneous (all monomials have the same number of
variables), tropical circuits are no more powerful than computing  $\{+,\prd\}$-circuits.
Still, for non-homogeneous polynomials, tropical circuits can be exponentially
more powerful than even counting $\{+,\prd\}$-circuits.
In fact, both gaps $\C(f)/\T(f)$ and $\T(f)/\C(f)$ can be exponential,
meaning that tropical and counting $\{+,\prd\}$-circuits are incomparable.

\begin{prop}
There are multilinear polynomials $f$ and $g$ of $n$ variables such that
both $\C(f)/\T(f)$ and $\T(g)/\C(g)$ are $2^{\Omega(\sqrt{n})}$.
\end{prop}

\begin{proof}
To show the first gap, consider the permanent polynomial
$f=\PERM{m}+\sum_{i,j=1}^m x_{ij}$ on $n=m^2$ variables. Theorem~\ref{thm:main01} gives
$\C(f)=2^{\Omega(m)}$. But $\T(f)\leq m^2=n$ because $f$ can be computed by
a tropical circuit $F=\sum_{i,j} x_{ij}$ whose tropicalization is $\trop{F}=\min_{i,j}(x_{ij})$: since variables cannot take negative values, the minimum
will be achieved on a single variable. Thus, $\C(f)/\T(f)=2^{\Omega(m)}$.

To show the second gap, take the multilinear  polynomial $g$ considered in the proof of
Theorem~\ref{thm:main0}.
The polynomial $g$ is the linearization of the polynomial $\ISOL{n}$ on $n=m^2$ variables given by \eqref{eq:isol}, and has $\C(g)=O(n^{3/2})$. On the other hand,
every monomial of  $g$ has degree between $m$ and
$2m$, and the monomials of degree $m$ correspond to perfect
matchings. Thus, the lower envelope $\lenv{g}$ of $g$ is just the permanent polynomial, i.e. $\lenv{g}=\PERM{m}$. Since $\CC(\PERM{m})\geq \C(\PERM{m}) = 2^{\Omega(m)}$
(see Corollary~\ref{cor:bounds1}) and
$\T(g)\geq \CC(\lenv{g})$, the desired lower bound $\T(g)=2^{\Omega(m)}$ follows.
\end{proof}

As we mentioned above, $\T(f)\geq \CC(\lenv{f})$ holds for  every multilinear polynomial $f$. Thus,
if the lower envelope $\lenv{f}$ requires large monotone arithmetic circuits, then
the polynomial $f$ itself requires large tropical circuits. This, however, does not hold for polynomials whose lower envelopes have small $\{+,\prd\}$-circuits.
 An important example in this respect is the $s$-$t$ path polynomial $f=\PATH{n}$.
Even though we have $\CC(f)=2^{\Omega(n)}$ \cite{jerrum},
the lower envelope of $f$  consist of just one variable $x_{s,t}$, implying that
$\CC(\lenv{f})=0$. And indeed, the Bellman--Ford algorithm (see Sect.~\ref{sec:reduction}) gives $\T(f)=O(n^3)$.

\begin{probl}
Does $\T(f)=\Omega(n^3)$ hold for $f=\PATH{n}$?
\end{probl}
This would show that the Bellman--Ford algorithm is optimal, if only Min and Plus operations can be used.
It is worth to mention that the optimality of the other prominent dynamic programming algorithm---that of
Floyd--Warshall \cite{floyd,warshall} for the all-pairs shortest paths problem---is already known. The corresponding
to this problem ``polynomial'' $\ALLPATH{n}$ is actually a \emph{set} of $s$-$t$
path polynomials $\PATH{n}$ for all choices of the source and target nodes $s$ and $t$.
Thus, unlike for $\PATH{n}$, every circuit for $\ALLPATH{n}$ must already have $\Omega(n^2)$ distinct output gates.
The Ford--Warshall algorithm gives $\T(\ALLPATH{n})=O(n^3)$. On the other hand,
Kerr~\cite{kerr} has shown that also $\T(\ALLPATH{n})=\Omega(n^3)$ holds.

 In Sect.~\ref{sec:reduction}, we have shown that the monotone counting complexity of $\PATH{n}$ is exponential in $n$.
 But, unlike bounds given in Sect.~\ref{sec:appl}, our proof for $\PATH{}$ indirect and is based on two rather non-trivial known results:
 the fact that the clique polynomial $\Clique{}$ requires exponential monotone real circuits, and is a projection of the Hamiltonian $s$-$t$ path polynomial $\HP{}$.

\begin{probl}
 Give a direct proof of
$\C(f)=2^{n^{\Omega(1)}}$ for $f=\PATH{n}$.
\end{probl}

Finally, it would be interesting to extend to the case of counting $\{+,\prd\}$-circuits one of the first lower-bounds arguments for computing $\{+,\prd\}$-circuits suggested by
Schnorr in \cite{schnorr}. Namely, he proved that $\CC(f)\geq |\mon{f}|-1$ holds, if the
polynomial $f$ is \emph{separated} in the following sense: for every two monomials
$p\neq q$ of $f$, their product $p\pr q$ does not contain any third monomial $r\not\in\{p,q\}$ of $f$ as a factor (see also \cite[Sect. 8]{juk14} for a somewhat simpler  proof). This criterion allows to easily prove strong lower bounds for some polynomials.
For example, using it, one can easily show that $\CC(f)\geq \tbinom{n}{r}-1$ holds for the $r$-clique polynomial $f=\Clique{n,r}$. This polynomial is separated, because
the union of no two $r$-cliques (sets of edges of complete subgraphs of $K_n$ with $r$ nodes) can contain a third $r$-clique.

\begin{probl}
Can Schnorr's argument for $\CC(f)$ be extended to $\C(f)$?
\end{probl}

\subsection*{Acknowledgments}

I am thankful to Tsuyoshi Ito, Emil Je\v{r}\'abek, and
Igor Sergeev for interesting discussions.

\footnotesize
\def\Bib#1{\vspace{-2pt}\bibitem{#1}}

\end{document}